\documentclass[12pt,reqno]{amsart}

\newcommand\version{February 28, 2015}

\usepackage{amsmath,amsfonts,amsthm,amssymb,amsxtra, amsbsy}
\usepackage{color}
\usepackage{hyperref}

\newtheorem{theorem}{Theorem}[section]
\newtheorem{proposition}[theorem]{Proposition}
\newtheorem{lemma}[theorem]{Lemma}
\newtheorem{corollary}[theorem]{Corollary}

\theoremstyle{definition}

\theoremstyle{remark}

\newtheorem{remark}[theorem]{Remark}

\numberwithin{equation}{section}

\renewcommand{\epsilon}{\varepsilon}

\newcommand{\N}{\mathbb{N}}

\renewcommand{\phi}{\varphi}
\newcommand{\R}{\mathbb{R}}

\DeclareMathOperator{\diam}{diam}

\DeclareMathOperator{\Div}{div}

\DeclareMathOperator{\per}{Per}

\begin{document}

\title[Existence of minimizers --- \version]{A compactness lemma and its application to the existence of minimizers for the liquid drop model}

\author{Rupert L. Frank}
\address{Rupert L. Frank, Mathematics 253-37, Caltech, Pasadena, CA 91125, USA}
\email{rlfrank@caltech.edu}

\author{Elliott H. Lieb}
\address{Elliott H. Lieb, Departments of Mathematics and Physics, Princeton
University, Princeton, NJ 08544, USA}
\email{lieb@princeton.edu}

\begin{abstract}
The
ancient Gamow liquid drop model of nuclear energies has had a renewed life
as an interesting problem in the calculus of
variations: Find a set $\Omega \subset \R^3$ with given volume A that
minimizes the sum of its surface area and its Coulomb self energy. A ball
minimizes the former and maximizes the latter, but the conjecture is that a 
ball is always a minimizer -- when there is a minimizer. Even the existence
of minimizers for this interesting geometric problem has not been shown
in general. We 
prove the existence of the absolute minimizer (over all $A$) of the energy
divided by $A$ (the binding energy per particle). A second result of our
work is a general method for 
showing the existence of optimal sets in geometric minimization problems,
which we call the `method of the missing mass'. A third point is the
extension of the pulling back compactness lemma \cite{LiIntersection}
from $W^{1,p}$ to $BV$. 
\end{abstract}


\maketitle

\renewcommand{\thefootnote}{${}$} \footnotetext{\copyright\, 2015 by
  the authors. This paper may be reproduced, in its entirety, for
  non-commercial purposes.\\
  Work partially supported by U.S. National Science
Foundation PHY-1347399 and DMS-1363432 (R.L.F.) and PHY-1265118. (E.H.L.)}


\section{Introduction and main results}

In this paper we are interested in some aspects of Gamow's famous 1930
`liquid drop model'  \cite{Ga}. 
This simple, but successful model of
atomic nuclei predicts, among other things, 
\begin{enumerate} 
\item[(a)]
the spherical shape of nuclei,
\item[(b)] the non-existence of nuclei with atomic number larger than some critical
value, and
\item[(c)] the existence of a nucleus with minimal binding energy per particle.
\end{enumerate}
Rather surprisingly, this model has received little attention in
mathematics, and only recently have proofs appeared which rigorously derive
properties (a) (partially) and (b). One of our goals in this paper will
be to prove property (c). Before we describe these results in detail,
we recall the model.

In the liquid drop model a nucleus $\Omega\subset\R^3$ is assumed to
have constant density, which we may assume to be equal to one. Thus,
$|\Omega|=A$ is the number of nucleons (protons and neutrons) in the
nucleus. Mathematically, it is not necessary to assume that this number
is an integer. The binding energy of a nucleus is given by $$ \mathcal
E(\Omega) := \per\Omega + D(\Omega) \,, $$ where $\per\Omega$ denotes
the surface area of $\Omega$, provided its boundary is smooth (see
\eqref{eq:perimeter} for a definition of $\per\Omega$ for an arbitrary
measurable set), and $$ D(\Omega) := \frac{1}{2} \iint_{\Omega\times\Omega}
\frac{dx\,dy}{|x-y|} $$ denotes the Coulomb repulsion energy. We dropped a
volume term $-eA$, where $e$ is the energy of `nuclear matter'
per unit volume in the `thermodynamic limit'. (This energy is computed using only the short-range strong forces, without the Coulomb repulsion.) This term 
only contributes a constant in our case, and need not be exhibited
explicitly for our considerations, but stability of a real nucleus
requires its inclusion; the total energy must be negative. The surface
term $\per\Omega$ accounts for corrections to the volume term coming
from the surface of the nucleus and acts as a surface tension. The
Coulomb term $D(\Omega)$ describes the repulsion between the protons
in the nucleus. The charge (which we set equal to one in the definition
of $D(\Omega)$) is an effective charge proportional to the ratio of the
proton number to $A$. We also neglected asymmetry and pairing terms.

The minimal binding energy of a nucleus with mass number $A>0$ is given by
$$
E(A) := \inf\{ \mathcal E(\Omega):\ |\Omega|=A \} \,.
$$
The following is rigorously known about the properties (a) and (b) mentioned before. 
\begin{enumerate}
\item[(a)] There is an $A_{c_1}>0$ such that $E(A)$ is minimized by a ball
for all $0<A\leq A_{c_1}$ and such that a ball does not minimize
$E(A)$ for $A>A_{c_1}$. It is even known that balls are the \emph{unique}
minimizers for $A<A_{c_1}$. 
\item[(b)] There is an $A_{c_2}>0$ such that $E(A)$ has no minimizer for
$A>A_{c_2}$.
\end{enumerate}
Fact (a) was proved by Kn\"upfer and Muratov \cite[Thm. 3.2]{KnMu} (balls are minimizers for
$A\leq A_{c_1}$; see also \cite{Ju}) and by Bonacini and Cristoferi \cite[Thm. 2.10]{BoCr}
(balls are unique minimizers for $A<A_{c_1}$ and are not minimizers for
$A>A_{c_1}$). Fact (b) is proved by Lu and Otto \cite[Thm. 2]{LuOt} and their 2014
presentation of this work at the Fields Institute was partial motivation
for our interest in the problem.

Fact (b) has the physical interpretation that nuclei `fission' spontaneously when they are too large.  When the mass is slightly below the critical value $A_{c_2}$ they are still unstable against being impacted by a neutron -- with the release of some notable amount of energy. In this model the reduction in energy is electrostatic (Coulomb), not nuclear, so one might properly refer to a `Coulomb bomb' rather than to a `nuclear bomb'.

It is \emph{conjectured} \cite{ChPe}, but not known, that $A_{c_1}=A_{c_2}$.

In this paper we show that (see Theorem \ref{iron}):
\begin{enumerate}
\item[(c)] There is a set $\Omega_*\subset\R^3$ such that
$$
\frac{E(A)}{A} \geq \frac{\mathcal E(\Omega_*)}{|\Omega_*|}
$$
for any $A>0$.
\end{enumerate}
Physically, the set $\Omega_*$ corresponds to an iron (Fe)
nucleus which is the element with the greatest binding energy per particle. Our proof of (c) makes use of the Lu--Otto analysis in \cite{LuOt}. From
(c) one can easily deduce that (see Corollary \ref{tdlimit})
$$
\lim_{A\to\infty} \frac{E(A)}{A} = \frac{\mathcal E(\Omega_*)}{|\Omega_*|} \,.
$$

Recall from (b) that
$$
A_{c_2} := \sup\left\{ A>0:\ E(A) \ \text{has a minimizer} \right\} <\infty \,.
$$
We consider the question whether there is a minimizer for $E(A_{c_2})$. By
definition there is a sequence of sets $\Omega_n$ with $A_n:=|\Omega_n|\to
A_{c_2}$ which minimize $E(A_n)$. The question is how these sets behave as
$n\to\infty$. If they have a limit, this limit is a good candidate for a
minimizer for $E(A_{c_2})$. On the other hand, the fact that there is no
minimizer beyond $A_{c_2}$ means, in some sense, that the Coulomb energy
becomes the dominant term at $A_{c_2}$. In order to minimize the Coulomb
energy, one might think that the sets $\Omega_n$ become more and more
elongated, which could cause the failure of a non-trivial
limit. That this possible scenario does \emph{not} occur is our next result
(see Theorem \ref{lead}).
\begin{enumerate}
\item[(d)] $E(A_{c_2})$ has a minimizer.
\end{enumerate}
Physically, a corresponding minimizing set $\Omega_{c_2}$ corresponds to
lead (Pb), a stable nucleus with the largest possible nucleon number. Our proof of
(d) uses crucially bounds from \cite{LuOt} for the proof of (b).

We also prove the following (physically obvious) fact (see Theorem \ref{binding}).
\begin{enumerate}
\item[(e)] If for a given $A>0$ one has the strict binding inequality
$E(A)<E(A')+E(A-A')$ for all $0<A'<A$, then there is a minimizer for $E(A)$.
\end{enumerate}
We emphasize that the strict binding inequality is only a sufficient condition for
the existence of a minimizer. For instance in the situation of (d) it is probably
violated.

\medskip

A second goal in this paper, besides the study of the liquid drop model, is
to present a method to prove the existence of optimal sets in geometric
minimization problems that involve the perimeter functional. We believe that
our method is simpler and more direct than methods used, for instance, in
\cite{dPVe} and has applications in related problems.

The method that we are using may be called the {\it `method of
the missing mass'}  and originates in the solution of the HLS problem
\cite{LiHLS} and also used in the paper \cite[Lemma 1.2 and Appendix]{BrNi}.
In this method each
element of a minimizing sequence is decomposed into two terms, namely, a
main piece with good convergence properties and a missing piece which
vanishes in a weak sense. The key point is then to use this missing piece as
a potential minimizer. Typically, this allows one to conclude that the limit
of the main piece is a minimizer. This method is often, but not in this
paper, combined with a strengthening of Fatou's lemma which includes the
missing piece; see \cite{LiHLS} and, for an extension to a wide class of
convex functions, see \cite{BrLi1}. (The latter paper also contains a review
of how this method is used in \cite{LiHLS} and \cite{BrNi}.)

The method of the missing mass works both in problems with and without
trans\-lation invariance. The minimization problems in \cite{BrNi} is not
translation invariant and the translation invariance in \cite{LiHLS} is
broken by symmetric decreasing rearrangement. In situations \emph{with}
translation invariance, however, as in the present paper, there is the
additional difficulty of detecting a main piece which has a non-zero
limit. This problem was overcome in \cite{LiIntersection}; see
\cite{BrLi2,FrLiLo} for early applications of this pulling back compactness lemma to concrete minimization problems. More recent applications include
\cite{FrLi,BeFrVi}.

As far as we know, the method of the missing mass has not been applied to geometric minimization problems in which the minimization is over sets. In this paper we identify a possible non-zero limit after translations in much the same way as in \cite{LiIntersection}, where the case of the Sobolev space $W^{1,p}(\R^d)$ with $p>1$ was treated. Our case here corresponds to $p=1$, where the space $BV(\R^d)$ of functions of bounded variation has to be used to get better compactness properties; see Proposition \ref{runningafter}. What is considerably more involved than in the $p=2$ case is to split the elements of the minimizing sequence into a main and a missing piece. This is accomplished in Lemma \ref{dichotomy} and the method used has some similarities with the one used in \cite{FrLiSeSi} in a different physical context.


\section{Compactness up to translations}\label{sec:exabstr}

In this section we describe sequences of subsets of $\R^d$ with a uniform perimeter bound and we prove a pulling back compactness theorem, which extends that of \cite[Sect. 3]{LiIntersection} to functions of bounded variation.

If $E_n$ and $E$ are measurable sets in $\R^d$, we say that $E_n \to E$ (globally)
if $|E_n\Delta E|\to 0$ and we say that $E_n\to E$ locally if for every compact set
$K\subset\R^d$,
$$
\left|\left(E_n\Delta E\right) \cap K\right| \to 0 \,.
$$
If $E\subset\R^d$ is a measurable set, we denote by
\begin{equation}
\label{eq:perimeter}
\per E = \sup\left\{ \int_E \Div F\,dx:\ F\in C^1_0(\R^d,\R^d)\,,\ |F|\leq 1 \right\}
\end{equation}
its perimeter in the sense of De Giorgi. Moreover, $B_r(a)$ denotes the open ball of
radius $r$ centered at $a\in\R^d$, and we abbreviate $B_r=B_r(0)$.

The following proposition is our technical main result. As explained in the introduction, its proof relies on a
technique from \cite{LiIntersection}.

\begin{proposition}\label{runningafter}
Let $(E_n)$ be a sequence of measurable sets in $\R^d$ with uniformly bounded
perimeter. Then one of the following two alternatives occurs:
\begin{enumerate}
\item $\lim_{n\to\infty} |E_n| =0$
\item There is a set $E$ with positive measure and a sequence $(a_k)$ in $\R^d$ such
that for a subsequence $(n_k)$ one has
$$
E_{n_k}-a_k \to E
\qquad\text{locally}.
$$
Moreover, $|E| \leq \liminf_{k\to\infty} |E_{n_k}|$ and $\per E \leq
\liminf_{k\to\infty} \per E_{n_k}$.
\end{enumerate}
\end{proposition}

\begin{proof}
We begin by proving that for every $u\in W^{1,1}(\R^d)$ and every $r>0$,
\begin{equation}
\label{eq:runningafterproof}
\int_{\R^d} \left( |\nabla u| + C r^{-1} |u| \right)dx \geq c \left( \sup_{a\in\R^d}
\left|B_r(a)\cap \{ u \neq 0\} \right| \right)^{-1/d} \int_{\R^d} |u|\,dx
\end{equation}
for some constants $c$ and $C$ depending only on $d$. In fact, let $\chi$ be a
smooth function with support in $B_1$ and $\|\chi\|_1=1$. We put $\chi_{a,r}(x) =
r^{-d} \chi((x-a)/r)$ and compute
$$
\int_{\R^d} |\nabla(\chi_{a,r}u)| \,dx \leq \int_{\R^d} \left(|\chi_{a,r}| |\nabla
u| + |\nabla\chi_{a,r}| |u|\right) dx \,. 
$$
On the other hand, by a Sobolev inquality and H\"older inequality,
\begin{align*}
\int_{\R^d} |\nabla(\chi_{a,r}u)| \,dx & \geq c \left( \int_{\R^d} |\chi_{a,r}
u|^{d/(d-1)} \,dx \right)^{(d-1)/d} \\
& \geq c \left( \sup_{a\in\R^d} \left|B_r(a)\cap \{ u \neq 0\} \right|
\right)^{-1/d} \int_{\R^d} |\chi_{a,r} u|\,dx \,.
\end{align*}
When we integrate both the upper and the lower bound on $\|\nabla(\chi_{a,r}
u)\|_1$ with respect to $a$ we obtain \eqref{eq:runningafterproof} with $C= \|\nabla\chi\|_1$.

By a standard argument \cite[Thm. 3.9]{AmFuPa} \eqref{eq:runningafterproof} extends
to $u\in BV(\R^d)$ and therefore, in particular, to $u=\chi_{E_n}$, where $(E_n)$
has uniformly bounded perimeter. For $r=1$, we find that
$$
\sup_{a\in\R^d} | B_1(a)\cap E_n| \geq \left( \frac{c |E_n|}{\per E_n + C |E_n|}
\right)^d \,.
$$
If we assume that alternative (1) does not occur, we may pass to a subsequence and
assume that $\inf_n |E_n| >0$. Thus the right side of the previous inequality is
bounded away from zero and so there are $a_n\in\R^d$ such that
$$
\inf_n |B_1 \cap (E_n-a_n)| = \inf_n |B_1(a_n)\cap E_n| >0 \,.
$$
We now apply the compactness result for sets of finite perimeter \cite[Thm.
3.39]{AmFuPa} to $E_n - a_n$. We infer that there is a set $E$ of finite perimeter
such that $E_n- a_n\to E$ locally. It is easy to see that local convergence implies
that $\liminf |E_n| = \liminf |E_n - a_n| \geq |E|$. Finally, lower semi-continuity
of the perimeter, that is, $\liminf \per E_n = \per (E_n-a_n) \geq \per E$, follows
easily from the variational definition of the perimeter.
\end{proof}

Next, we describe sequences of sets which converge locally but not globally, because
they lose measure in the limit. We show that such sets can be decomposed into a main
piece, which does converge globally, and a remainder piece which disappears from any
compact set.

\begin{lemma}\label{dichotomy}
Let $(E_n)$ be a sequence of measurable sets in $\R^d$ with finite perimeter such
that $E_n\to E$ locally for some set $E\subset\R^d$. Assume that
$0<|E|<\liminf_{n\to\infty} |E_n|$. Then there is a sequence $(r_n)$ in $(0,\infty)$
such that the sets
$$
F_n := E_n \cap B_{r_n} \,,
\qquad
G_n := E_n \cap \left( \R^d \setminus \overline{B_{r_n}} \right)
$$
satisfy
\begin{equation}
\label{eq:perdown}
\lim_{n\to\infty} \left( \per E_n - \per F_n - \per G_n \right) = 0
\end{equation}
and
\begin{equation}
\label{eq:limit}
F_n \to E
\ \text{(globally)}\qquad \text{and}\qquad
G_n \to \emptyset
\ \text{locally}.
\end{equation}
In particular,
\begin{equation}
\label{eq:limitmassper}
\lim_{n\to\infty}|F_n| = |E|
\qquad\text{and}\qquad
\liminf_{n\to\infty} \per F_n \geq \per E \,.
\end{equation}
\end{lemma}

\begin{proof}
Clearly, we have
$$
\per E_n \leq \per F_n + \per G_n \leq \per E_n +2 \sigma_{r_n}(\partial B_{r_n}
\cap E_n) \,,
$$
where $\sigma_r$ denotes surface measure on the sphere $\partial B_r$. Thus, for the
bound \eqref{eq:perdown} on the perimeter, we have to prove that
\begin{equation}\label{eq:dichoproof}
\sigma_{r_n}(\partial B_{r_n} \cap E_n)\to 0
\end{equation}
for a suitable choice of $r_n$.

In order to construct the $r_n$ we distinguish two cases. We first assume that $E$
is (essentially) bounded, say, $E\subset B_{R}$ for some $R>0$. Then $E_n\to E$
locally implies that $\epsilon_n := |E_n\cap (B_{2R}\setminus B_R)|\to 0$. We claim
that for every $n$ there is an $r_n\in [R,2R]$ with $\sigma_{r_n}(\partial
B_{r_n}\cap E_n) \leq \epsilon_n/R$, proving \eqref{eq:dichoproof}. In fact, if this
was not the case, we had, by integration in spherical coordinates,
$$
\epsilon_n < \int_R^{2R} \sigma_r(\partial B_r\cap E_n) \,dr = |E_n\cap
(B_{2R}\setminus B_R)| = \epsilon_n \,,
$$
which is a contradiction.

Now assume that $E$ is (essentially) unbounded and define $R_n$ by $|E_n\cap
B_{R_n}|= |E|$. (This is well defined if $|E_n|>|E|$, which we may assume after
discarding finitely many $E_n$'s.) We claim that $R_n\to\infty$. In fact, if we had
$R_{n_k}\to R_*$ for some $R_*<\infty$, then, using $E_n\to E$ locally,
$$
|E| = |E_{n_k}\cap B_{R_{n_k}}| = |E_{n_k}\cap B_{R_*}| + o(1) = |E\cap B_{R_*}| +
o(1)< |E| + o(1) \,,
$$
which is a contradiction. Now choose $R$ such that $|E\cap B_R|=|E|/2$. Then, since
$E_n\to E$ locally, for all sufficiently large $n$, $|E_n\cap B_R|\geq
|E|/4$. After discarding finitely many $n$ we may assume that $R_n> R$. We
claim that for every $n$ there is an $r_n\in [(R+R_n)/2,R_n]$ with
$\sigma_{r_n}(\partial B_{r_n}\cap E_n) \leq 3|E|/(2(R_n-R))$, proving \eqref{eq:dichoproof} once more. In fact, if this were not the
case, we would have, like before,
$$
\tfrac34 |E| < \int_{\tfrac{R+R_n}2}^{R_n} \sigma_r(\partial B_r\cap E_n) \,dr =
|E_n\cap (B_{R_n}\setminus B_{\tfrac{R+R_n}2})| \leq |E_n\cap B_{R_n}| - |E_n\cap
B_R| \leq \tfrac34 |E| \,,
$$
which is a contradiction. This completes the proof of \eqref{eq:perdown}.

It remains to prove \eqref{eq:limit} (which implies \eqref{eq:limitmassper})
and, to do so, we again distinguish two cases. Assume first
that $E$ is bounded and note that
\begin{equation}
\label{eq:dichoproofbdd}
|F_n\Delta E| = 2 |F_n\setminus E| + |E| - |F_n| \,.
\end{equation}
Since $F_n \supset E_n\cap B_R$, we have $|F_n|\geq |E_n\cap B_R|$ and, by local
convergence, $|E_n\cap B_R|\to |E\cap B_R|=|E|$. Thus, $\liminf |F_n|\geq |E|$. On
the other hand, we have $F_n\subset E_n\cap B_{2R}$ and therefore $|F_n\setminus E|
\leq |(E_n\setminus E)\cap B_{2R}|$. By local convergence, $|(E_n\setminus E)\cap
B_{2R}|\to 0$. Because of \eqref{eq:dichoproofbdd} we conclude that $|F_n\Delta
E|\to 0$. Finally, if $K$ is compact, then $|G_n\cap K| \leq |E_n\cap
(\R^d\setminus\overline{B_R})\cap K|$. By local convergence, the latter converges to
$|E\cap (\R^d\setminus\overline{B_R})\cap K|=0$, so we conclude that
$G_n\to\emptyset$ locally.

Now assume that $E$ is unbounded and note that
\begin{equation}
\label{eq:dichoproofunbdd}
|F_n\Delta E| = 2 |E\setminus F_n| + |F_n| - |E| \,.
\end{equation}
Since $F_n\subset E_n\cap B_{R_n}$, we have $|F_n|\leq |E_n\cap B_{R_n}|=|E|$. On
the other hand, we have $|E\setminus F_n| = |(E\setminus E_n)\cap B_{r_n}| +
|E\setminus B_{r_n}|$. By dominated convergence, since $r_n\to\infty$, we have
$|E\setminus B_{r_n}|\to 0$. Now given $\epsilon>0$ let $\rho>0$ such that $|E\cap
(\R^d\setminus \overline{B_\rho})| \leq \epsilon$. Then, if $r_n\geq\rho$,
$|(E\setminus E_n)\cap B_{r_n}| = |(E\setminus E_n)\cap B_\rho| + |(E\setminus
E_n)\cap B_{r_n}\cap (\R^d\setminus\overline{B_\rho})|$. We have $|(E\setminus
E_n)\cap B_\rho|\to 0$ by local convergence and $|(E\setminus E_n)\cap B_{r_n}\cap
(\R^d\setminus\overline{B_\rho})|\leq \epsilon$. All this proves that $|E\setminus
F_n| \to 0$. Because of \eqref{eq:dichoproofunbdd} we conclude that $|F_n\Delta
E|\to 0$. Finally, if $K$ is compact, then $K\cap (\R^d\setminus B_{r_n})=\emptyset$
for all sufficiently large $n$, so $|G_n\cap K|=0$ and we conclude that
$G_n\to\emptyset$ locally. This concludes the proof of the lemma.
\end{proof}

Clearly, the sets $F_n$ and $G_n$ from the previous proposition satisfy
$$
|F_n| + |G_n| = |E_n| \,.
$$
We next show that, at least asymptotically, some non-local functionals are also
additive with respect to this decomposition. For $0<\lambda<d$ let
$$
I_\lambda(E) = \frac12 \iint_{E\times E} \frac{dx\,dy}{|x-y|^\lambda} \,.
$$

\begin{lemma}\label{coulomb}
Let $F_n$ and $G_n$ be sequences of measurable sets in $\R^d$ with uniformly bounded
measure and $|F_n\cap G_n|=0$ for all $n$. Assume that
\begin{equation*}
F_n \to E
\ \text{(globally)}\qquad \text{and}\qquad
G_n \to \emptyset
\ \text{locally}
\end{equation*}
for some set $E$. Then for all $0<\lambda<d$,
\begin{equation}
\label{eq:coulomb1}
I_\lambda(F_n \cup G_n) = I_\lambda(F_n) + I_\lambda(G_n) + o(1) \,.
\end{equation}
and
\begin{equation}
\label{eq:coulomb2}
I_\lambda(F_n) = I_\lambda(E) + o(1) \,.
\end{equation}
\end{lemma}

\begin{proof}
As a preliminary to the proof we note that for an arbitrary set $F\subset\R^d$ of
finite measure, by a simple rearrangement inequality \cite[Thm. 3.4]{LiLo},
\begin{equation}
\label{eq:coulombprelim1}
\int_F \frac{dy}{|x-y|^\lambda} \leq \int_{F^*} \frac{dy}{|y|^\lambda} = C
|F|^{(d-\lambda)/d}
\end{equation}
with some explicit $C$ depending only on $d$ and $\lambda$. Moreover,
\begin{equation}
\label{eq:coulombprelim2}
\int_F \frac{dy}{|x-y|^\lambda} \to 0
\qquad\text{as}\ |x|\to\infty \,.
\end{equation}
(This follows from \cite[Thm. 2.2]{LiLo} by decomposing $|x|^{-\lambda}$ in
a short and a long range part.)  

We now turn to the proof of \eqref{eq:coulomb1}. We need to show that
$$
2I_\lambda(F_n,G_n):= \iint_{F_n\times G_n} \frac{dx\,dy}{|x-y|^\lambda} \to 0 \,.
$$
Using \eqref{eq:coulombprelim1} and $F_n\to E$ we find
$$
\left| I_\lambda(F_n,G_n) - I_\lambda(E,G_n) \right| \leq I_\lambda(F_n\Delta E,G_n)
\leq C|F_n\Delta E|^{(d-\lambda)/d} |G_n| \to 0 \,,
$$
so it is enough to prove $I_\lambda(E,G_n)\to 0$. Let $\epsilon>0$ be given and use
\eqref{eq:coulombprelim2} to find $R>0$ such that $\int_E |x-y|^{-\lambda}\,dy \leq
\epsilon$ for $|x|\geq R$. We decompose and bound
$$
I_\lambda(E,G_n) = I_\lambda(E,G_n\cap B_R) + I_\lambda(E,G_n\setminus B_R) 
\leq C |E|^{(d-\lambda)/d} |G_n\cap B_R| + \epsilon |G_n\setminus B_R| \,.
$$
Since $G_n\to\emptyset$ locally, we have $|G_n\cap B_R|\to 0$ and therefore
$I_\lambda(E,G_n)\to 0$, as claimed. The proof of \eqref{eq:coulomb2} is similar.
\end{proof}


\section{Application to the liquid drop model}

\subsection{Results}

In this section we return to the liquid drop model discussed in the
introduction. Thus, we assume $d=3$, $\lambda=1$ and abbreviate
$I_1(\Omega)=D(\Omega)$. We recall that the binding energy of a set $\Omega$
and the minimal binding energy with nucleon number $A$ were defined as
$$
\mathcal E(\Omega) = \per\Omega + D(\Omega)
\qquad\text{and}\qquad
E(A) = \inf\{ \mathcal E(\Omega):\ |\Omega|=A \} \,.
$$
It is easy to see (and proved, for instance, in \cite{LuOt}) that
\begin{equation}
\label{eq:bindingweak}
E(A) \leq E(A')+ E(A-A')
\qquad\text{for all}\ 0<A'<A \,.
\end{equation}
Our first theorem states that if this inequality is strict, then there is a
minimizer for $E(A)$.

\begin{theorem}\label{binding}
Let $A>0$ such that
\begin{equation}
\label{eq:bindingineq}
E(A) < E(A') + E(A-A')
\qquad\text{for all}\ 0<A'<A \,.
\end{equation}
Then the infimum defining $E(A)$ is attained. Moreover, any minimizing sequence has
a subsequence which, after a translation, converges (globally) to a minimizer.
\end{theorem}

We recall that (global) convergence of sets was defined at the beginning of
Section~\ref{sec:exabstr}. It means that the measure of the symmetric difference
between the sets and their limit tends to zero.

Of physical interest is the quantity
$$
e(A) := \frac{E(A)}{A} \,,
$$
the binding energy per particle.

\begin{theorem}\label{iron}
There is an $A_*>0$ such that $e(A_*)=\inf\{ e(A):\ A>0 \}$. Moreover, the infimum
defining $E(A_*)$ is attained.
\end{theorem}

In other words, there is a set $\Omega_*\subset\R^3$ with finite perimeter such that
\begin{equation}
\label{eq:iron}
\frac{\per\Omega + D(\Omega)}{|\Omega|} \geq \frac{\per \Omega_* +
D(\Omega_*)}{|\Omega_*|}
\qquad\text{for all}\ \Omega\subset\R^3 \ \text{with finite perimeter}.
\end{equation}

Theorem \ref{iron} has the following simple corollary.

\begin{corollary}\label{tdlimit}
As $A\to\infty$, $e(A)\to e(A_*)$. Moreover, $e(kA_*)=e(A_*)$ for all $k\in\N$.
\end{corollary}

\begin{proof}[Proof of Corollary \ref{tdlimit}]
Subadditivity \eqref{eq:bindingweak} and non-negativity of $E(A)$ imply by abstract
principles that $\lim_{A\to\infty} e(A)$ exists. (This is sometimes called `Fekete's
lemma'.) Iterating \eqref{eq:bindingweak} we infer that $e(kA_*) \leq e(A_*)$ for
all $k\in\N$. On the other hand, since $A_*$ is the global minimium of $e(A)$, we
have $e(kA_*)\geq e(A_*)$. So $e(kA_*)=e(A_*)$ for all $k\in\N$ and, since we know
that $e(A)$ has a limit, this limit must be equal to $e(A_*)$, as claimed.
\end{proof}

Our final theorem is

\begin{theorem}\label{lead}
The set $\{ A>0 :\ E(A) \ \text{has a minimizer} \}$ is closed.
\end{theorem}

This proves, in particular, that for $A_{c_2}= \sup \{ A>0 :\ E(A) \ \text{has a
minimizer} \}$ there is a minimizer.


\subsection{Proof of Theorems \ref{binding}, \ref{iron} and \ref{lead}}

We now prove our three main theorems about the liquid drop model.

\begin{proof}[Proof of Theorem \ref{binding}]
Let $(\Omega_n)$ with $|\Omega_n|=A$ be a minimizing sequence for $E(A)$. Since
$\per\Omega_n \leq \mathcal E(\Omega_n)= E(A)+o(1)$ is uniformly bounded,
Proposition \ref{runningafter} yields a set $\Omega\subset\R^3$ with $0<|\Omega|\leq
A$ and $\per\Omega\leq \liminf \per\Omega_n$ such that, after passing to a
subsequence and a translation, we have $\Omega_n\to\Omega$ locally. Moreover, by
Fatou's lemma, we have $D(\Omega)\leq\liminf D(\Omega_n)$, so $\mathcal E(\Omega)
\leq \liminf \mathcal E(\Omega_n)$. Thus, $\Omega$ will be a minimizer provided we
can show that $|\Omega|=A$. Moreover, it is easy to see that $\Omega_n\to\Omega$
locally and $|\Omega_n|\to |A|$ implies that $\Omega_n\to\Omega$ globally. Therefore both statements of the theorem follow if we can prove that $|\Omega|=A$.

We argue by contradiction and assume that $|\Omega|<A$. Then according to
Lemma~\ref{dichotomy} we can write $\Omega_n= F_n \cup G_n$, $F_n\cap
G_n=\emptyset$, such that
$$
\per \Omega_n \geq \per \Omega + \per G_n + o(1) \,.
$$
Moreover, $F_n\to\Omega$ globally and $G_n\to\emptyset$ locally which, by Lemma
\ref{coulomb}, implies that
$$
D(\Omega_n) = D(\Omega) + D(G_n) + o(1) \,.
$$
Thus,
\begin{align*}
E(A) & = \mathcal E(\Omega_n)+ o(1) \\
& \geq \mathcal E(\Omega) + \mathcal E(G_n) + o(1) \\
& \geq E(|\Omega|) + E(|G_n|) + o(1) \,.
\end{align*}
Since $|G_n|=|\Omega_n|-|F_n|\to A-|\Omega|$ and since $A\mapsto E(A)$ is continuous
on $(0,\infty)$ (in fact, $A^{-2/3} E(A)= \inf\{ \per\omega + A D(\omega):\
|\omega|=1 \}$ is concave as an infimum over affine linear functions), we obtain in
the limit $n\to\infty$
$$
E(A) \geq E(|\Omega|) + E(A-|\Omega|) \,,
$$
which contradicts assumption \ref{eq:bindingineq}. Thus, $|\Omega|=A$.
\end{proof}

We prove Theorem \ref{iron} via the auxiliary minimization problem
$$
e_\leq(A) := \inf\left\{ \frac{\mathcal E(\Omega)}{|\Omega|}:\ 0<|\Omega|\leq A
\right\} \,.
$$
The idea of relaxing the equality constraint $|\Omega|=A$ to $|\Omega|\leq A$ is reminiscent of \cite{LiSi}.

\begin{lemma}\label{relax}
For any $A>0$, the infimum defining $e_\leq(A)$ is attained.
\end{lemma}

\begin{proof}
Let $(\Omega_n)$ with $|\Omega_n|\leq A$ be a minimizing sequence for $e_\leq(A)$. By the isoperimetric inequality, we have $e_\leq(A)+o(1) = \mathcal E(\Omega_n)/|\Omega_n| \geq C |\Omega_n|^{-1/3}$, which implies that $\liminf |\Omega_n|>0$. Moreover, $\per\Omega_n\leq \mathcal E(\Omega_n) \leq A(e_\leq(A)+o(1))$ is uniformly bounded, so as in the proof of Theorem \ref{binding} we obtain, after passing to subsequence and after a translation, a set $\Omega\subset\R^3$ with $0<|\Omega|\leq A$ and with $\Omega_n\to\Omega$ locally and $\mathcal E(\Omega) \leq\liminf\mathcal E(\Omega_n)$. We now distinguish two cases, according to how $|\Omega|$ compares to $\liminf |\Omega_n|$.

If $|\Omega|\geq\liminf|\Omega_n|$, then $\mathcal E(\Omega)/|\Omega| \leq \limsup \left( \mathcal E(\Omega_n)/|\Omega_n| \right) = e_\leq(A)$, so $\Omega$ is a minimizer for $e_\leq(A)$.

If $|\Omega|<\liminf |\Omega_n|$, then, as in the proof of Theorem \ref{binding},
\begin{align*}
e_\leq(A) & = \frac{\mathcal E(\Omega) + \mathcal E(G_n)}{|\Omega|+|G_n|} + o(1)
\geq \frac{\mathcal E(\Omega)}{|\Omega|+|G_n|} + \frac{|G_n|}{|\Omega|+|G_n|}\,
e_\leq(A) + o(1) \,.
\end{align*}
Rearranging the terms, we obtain $|\Omega|\, e_\leq (A) \geq \mathcal E(\Omega) + o(1)$, which again means that $\Omega$ is a minimizer for $e_\leq(A)$.
\end{proof}

\begin{proof}[Proof of Theorem \ref{iron}]
Clearly, $A\mapsto e_\leq(A)$ is non-increasing and non-negative, and it is continuous because of the continuity of $E\mapsto E(A)$ (see the proof of Theorem \ref{binding}). Let
$$
A_* := \sup\left\{ A>0:\ \text{there is an}\ A'>A \ \text{with}\ e_\leq(A')< e_\leq(A) \right\} \,.
$$
We claim that $A_*<\infty$. Clearly, this implies that $e_\leq(A)=e_\leq(A_*)$ for all $A\geq A_*$ and therefore that $e(A)\geq e(A_*)$ for all $A>0$. Moreover, the minimizer $\Omega_*$ for $e_\leq(A_*)$, which exists by Lemma \ref{relax}, has $|\Omega_*|=A_*$ and therefore is also a minimizer for $E(A_*)$.

We argue by contradiction and assume that $A_*=\infty$, that is, there is an increasing sequence $(A_n)$ with $A_n\to\infty$ such that $e_\leq(A_{n+1})<e_\leq(A_n)$ for all $n$. According to Lemma \ref{relax} there are sets $\Omega_n$ with $\mathcal E(\Omega_n) = e_\leq(A_n) |\Omega_n|$ and $|\Omega_n|\leq A_n$. The strict inequality $e_\leq(A_{n+1})<e_\leq(A_n)$ implies that $|\Omega_{n+1}|> A_n$ and so, in particular, $A_n':=|\Omega_n|\to\infty$. Since the sets $\Omega_n$ minimize $e_\leq(A_n')$, they also minimize $E(A_n')$, but the existence of minimizers for this problem with arbitrarily large $A_n'$ contradicts the result in \cite[Thm. 2]{LuOt}. This proves that $A_*<\infty$, as claimed.
\end{proof}

\begin{remark}
The minimization problem of Theorem \ref{iron} is equivalent to the following scale-invariant minimization problem,
$$
I = \inf\left\{ \frac{(\per\Omega)^{2/3} D(\Omega)^{1/3}}{|\Omega|} :\ \Omega\subset\R^3 \ \text{of finite perimeter} \right\}
$$
(We know from \cite[Lem. 7.1]{KnMu} that $I>0$.) In fact, in order to minimize $\mathcal E(\Omega)/|\Omega|$ we can minimize separately over shape and size of $\Omega=\ell\omega$, that is,
$$
\inf_\Omega \frac{\mathcal E(\Omega)}{|\Omega|} = \inf_{|\omega|=1} \inf_{\ell>0}
\left( \ell^{-1} \per\omega + \ell^2 D(\omega) \right) \,.
$$
For fixed $\omega$, the infimum is attained at
$\ell_\omega=(\per\omega/(2D(\omega))^{1/3}$ and we have
$$
\inf_\Omega \frac{\mathcal E(\Omega)}{|\Omega|} = 2^{-2/3} \cdot 3 \cdot
\inf_{|\omega|=1} (\per\omega)^{2/3} D(\omega)^{1/3} = 2^{-2/3} \cdot 3 I \,.
$$
In particular, if $\omega$ has $|\omega|=1$ and minimizes $(\per\Omega)^{2/3}
D(\omega)^{1/3}$, then $\Omega:=\ell_\omega \omega$ minimizes $\mathcal
E(\Omega)/|\Omega|$.
\end{remark}

\begin{proof}[Proof of Theorem \ref{lead}]
Let $A_n$ be a sequence with $A_n\to A\in (0,\infty)$ such that $E(A_n)$ has a
minimizer $\Omega_n$ with $|\Omega_n|=A_n$. As observed in the proof of Theorem
\ref{binding}, $A\mapsto E(A)$ is continuous, so $E(A_n)\to E(A)$ and $\per\Omega_n$
is uniformly bounded. As in the proof of Theorem \ref{binding} we may assume that
$\Omega_n\to\Omega$ locally for a set $\Omega\subset\R^3$ with $0<|\Omega|\leq A$
and $\mathcal E(\Omega)\leq\liminf \mathcal E(\Omega_n)$, and it remains to prove
that $|\Omega|=A$.

To prove this, we use a bound (implicitly) contained in \cite{LuOt}. Namely, for
every $\epsilon>0$ there is a constant $C_\epsilon>0$ such that if $A'\geq\epsilon$
and $\Omega'\subset\R^3$ is a minimizer for $E(A')$, then
\begin{equation}
\label{eq:leadproof2}
\diam\Omega'\leq C_\epsilon A \,.
\end{equation}
This is stated in \cite[Lemma 7.2]{KnMu} with $\epsilon=1$, but the same proof works
for $\epsilon<1$. (Technically speaking, in order to avoid ambiguities with sets of
measure zero, $\Omega'$ is replaced by the set $\{ x\in\R^3:\ \limsup_{r\to 0}
|\Omega'\cap B_r(x)|/|B_r(x)|>0 \}$. By Lebesgue's differentiation theorem, this set coincides with $\Omega'$ up to sets of measure zero, so none of the terms in the minimization problem changes under this replacement.)

Applying \eqref{eq:leadproof2} to $\Omega'=\Omega_n$ we infer that
$\limsup_{n\to\infty} \diam\Omega_n <\infty$. Thus, the sets $\Omega_n$ are
contained in a fixed ball, and then local convergence $\Omega_n\to\Omega$ implies
global convergence, so $|\Omega|=\lim |\Omega_n| = \lim A_n =A$, as claimed.
\end{proof}

\begin{remark}
Theorems \ref{binding} and \ref{lead} remain true, mutatis mutandis, for the functional $\mathcal E(\Omega) = \per\Omega + I_\alpha(\Omega)$, $\Omega\subset\R^d$, with $d\geq 2$ and $0<\alpha<d$. (For the analogue of the diameter bound \eqref{eq:leadproof2} see \cite[Lemma 7.2]{KnMu}.) Theorem \ref{iron} and Corollary \ref{tdlimit} remain true under the additional assumption $0<\alpha<2$. (Under this assumption the analogue of the Lu--Otto non-existence result remains valid, see \cite[Thm. 3.3]{KnMu}.)
\end{remark}


\section{Further remarks about the liquid drop model}

\subsection{A lower bound on $A_*$}

Let us consider the binding energy per particle for balls,
\begin{equation}
\label{eq:bindingball}
e^{(ball)}(A) := A^{-1/3} \,\frac{\per B}{|B|^{2/3}} + A^{2/3} \,\frac{D(B)}{|B|^{5/3}} \,,
\end{equation}
where $B$ is the unit ball in $\R^3$. (So $|B|=4\pi/3$, $\per B=4\pi$ and $D(B)$ can
be computed, but we don't need this.) Clearly, $e^{(ball)}(A)$ decreases up to some
$A_*^{(ball)}>0$ and then increases. Setting the $A$-derivative of $e^{(ball)}(A)$
equal to zero we find
\begin{equation}
\label{eq:astarball}
A_*^{(ball)} = \frac{|B|\,\per B}{2\, D(B)} \,.
\end{equation}
The main result of this section is the following quantitative bound on the curve
$e(A)$.

\begin{proposition}
Let
$$
A_0 := \sup\left\{ A_0>0:\ A\mapsto e(A)\ \text{is strictly decreasing on}\ (0,A_0)
\right\}
$$
Then $A_0 \geq A_*^{(ball)}$ with equality iff $E(A_0)$ is minimized by a ball.
\end{proposition}

Note that this gives, in particular, the lower bound $A_*\geq A_*^{(ball)}$ on the
number $A_*$ from Theorem \ref{iron}.

\begin{proof}
We know from \cite{KnMu} that $E(A)$ is minimized by balls for small $A$ and
therefore $e(A)$ is decreasing for small $A$, so $A_0>0$. From Theorem \ref{iron} we
know that $e(A)$ has a global minimum, so $A_0<\infty$. By assumption we have
$e(A)>e(A_0)$ for all $0<A<A_0$ and, therefore, 
$$
E(A') + E(A_0-A') > A' e(A_0) + (A_0-A') e(A_0) = E(A_0)
\qquad\text{for all}\ 0<A'<A_0 \,.
$$
According to Theorem \ref{binding} this implies that there is a minimizer $\Omega_0$
for $E(A_0)$.

The next step is to derive a \emph{virial relation} for $\Omega_0$. Since $e(A)$ has
a local minimum at $A_0$, the function
$$
\ell\mapsto \frac{\mathcal E(\ell\Omega_0)}{|\ell\Omega_0|} = \ell^{-1}\,
\frac{\per\Omega_0}{|\Omega_0|} + \ell^2\, \frac{D(\Omega_0)}{|\Omega_0|}
$$
has a local minimum at $\ell=1$. Setting the derivative at $\ell=1$ equal to zero we
conclude that
\begin{equation}
\label{eq:virial}
\per\Omega_0 = 2 D(\Omega_0) \,.
\end{equation}

We now use \eqref{eq:virial} to prove a lower bound on $A_0$. For the proof we will
use the inequalities
\begin{equation}
\label{eq:rearrange}
\per E \geq \per E^*
\qquad\text{and}\qquad
D(E^*) \geq D(E) \,, 
\end{equation}
where $E^*$ denotes a ball of the same measure as $E$. The first inequality is just
a rewriting of the isoperimetric inequality and the second one follows, for
instance, from the Riesz rearrangement inequality \cite[Thm. 3.7]{LiLo}. From these
inequalities and \eqref{eq:virial} we deduce that
$$
D(\Omega_0) = \frac12 \per\Omega_0 \geq \frac 12 \per\Omega_0^* = \frac12 \frac{\per
B}{D(B)^{2/5}} D(\Omega_0^*)^{2/5} \geq \frac12 \frac{\per B}{D(B)^{2/5}}
D(\Omega_0)^{2/5} \,.
$$
(The middle equality here just uses scaling.) Thus,
$$
D(\Omega_0)^{3/5} \geq \frac12 \frac{\per B}{D(B)^{2/5}} \,.
$$
On the other hand, again by \eqref{eq:rearrange},
$$
D(\Omega_0)^{3/5} \leq D(\Omega_0^*)^{3/5} = \frac{D(B)^{3/5}}{|B|} |\Omega_0^*| =
\frac{D(B)^{3/5}}{|B|} |\Omega_0| \,.
$$
(The middle equality again just uses scaling.) Combining the last two inequalities
and recalling \eqref{eq:astarball} we obtain
$$
|\Omega_0| \geq \frac{|B|\, \per B}{2\, D(B)} = A_*^{(ball)} \,,
$$
as claimed.

Finally assume that $\Omega_0$ is not a ball. Then the \emph{strict} rearrangement
inequality $D(\Omega_0^*)>D(\Omega_0)$ \cite{LiChoquard} implies, by the same argument as
before, the strict inequality $|\Omega_0| > A_*^{(ball)}$. This concludes the proof.
\end{proof}

\begin{remark}
The same proof shows that the virial relation \eqref{eq:virial} holds whenever $A_0$
is a local minimum of $e(A)$ and $E(A_0)$ is attained by some $\Omega_0$.
\end{remark}


\subsection{Dissociation into balls}

Let us consider
$$
\tilde e(A) := A^{-1} \inf\left\{ \mathcal E(\Omega):\ |\Omega|=A \,,\ \Omega\
\text{is a countable union of disjoint balls} \right\} \,.
$$
This is the same as the energy  $e(A)$ per particle except that we restrict the
allowed sets to be countable disjoint unions of balls. Clearly, it suffices to
consider the case when the individual balls are infinitely far apart and therefore
\begin{align*}
\tilde e(A) = \inf \left\{ \sum_{k=1}^\infty \frac{A_k}{A}\  e^{(ball)}(A_k) :\
\sum_{k=1}^\infty A_k = A \right\}
\end{align*}
with the energy $e^{(ball)}(A_k)$ from \eqref{eq:bindingball}. It turns out that
this minimization problem can be almost explicitly solved.

\begin{proposition}\label{dissoc}
For any $A>0$,
$$
\tilde e(A) = \min_{k\in\N} e^{(ball)}(A/k) \,.
$$
\end{proposition}

That is, $\tilde e(A)$ coincides with the energy per particle of $K$ infinitely
separated balls of volume $A/K$, where $K$ is to be optimized over. The energy per
particle of such a ball satisfies $e^{(ball)}(A/K)\sim (|\partial B|
/(A^{1/3}|B|^{2/3})) K^{1/3}$ as $K\to\infty$ and therefore the infimum over $K$ is
indeed attained at some finite $K$. With a little more work one can compute explicit
numbers $0=a_0<a_1<a_2<a_3<...$ with $a_k\to\infty$ such that $\tilde e(A)=
e^{(ball)}(A/k)$ for $a_{k-1}\leq A\leq a_k$. For instance,
$$
a_1 = \frac{2-2^{2/3}}{2^{2/3}-1}\ \frac{|B|\,\per B}{D(B)} \,.
$$
The proof of Proposition \ref{dissoc} is based on the following technical lemma.

\begin{lemma}\label{dissoclemma}
Let $\lambda>0$ and $f(\theta) = (\theta^{2/3} + (1-\theta)^{2/3}) + \lambda
(\theta^{5/3} + (1-\theta)^{5/3})$. Then for all $\theta\in [0,1]$
$$
f(\theta) \geq \min\{ f(0),f(1/2),f(1) \}
$$
with strict inequality unless $\theta\in\{0,1/2,1\}$.
\end{lemma}

\begin{proof}[Proof of Lemma \ref{dissoclemma}]
Differentiating we find
\begin{equation}
\label{eq:dissoclemmaproof}
f'(\theta) = \frac{5\lambda}{3} \left( \theta^{-1/3} - (1-\theta)^{-1/3} \right)
\left( \frac{2}{5\lambda} - g(\theta) \right) \,,
\end{equation}
where
$$
g(\theta) := \frac{(1-\theta)^{2/3}-\theta^{2/3}}{\theta^{-1/3}- (1-\theta)^{-1/3}} \,.
$$
Clearly, $g$ is symmetric with respect to $\theta=1/2$ and satisfies $g(0)=g(1)=0$.
By straightforward differentiation we find that $g$ is strictly concave, so it
attains its maximum at $\theta=1/2$ and we compute $g(1/2)=1$. Thus for every
$0\leq\gamma<1$ there is a unique $\theta(\gamma)\in[0,1/2)$ such that
$g(\theta)=\gamma$ iff $\theta=\pm\theta(\gamma)$.

Returning with this information to \eqref{eq:dissoclemmaproof} we find that, if
$2/(5\lambda)\geq 1$, then $f'$ has a single zero at $\theta=1/2$ and, if
$2/(5\gamma)<1$, then $f'$ has three zeros $\theta=1/2,\pm \theta(2/(5\lambda))$. Since $f$ is clearly increasing near $0$ and decreasing near $1$, the single zero of $f'$ at $\theta=1/2$ for $2/(5\lambda)\geq 1$ and the two zeros at
$\theta=\pm\theta(2/(5\lambda))$ for $2/(5\lambda)<1$ must correspond to local maxima
of $f$. Thus, $f$ attains its minimum at $\theta\in\{0,1\}$ for $2/(5\lambda)\geq
1$. For $2/(5\lambda)<1$, $f$ has a local minimum at $\theta=1/2$ and it attains its
minimum on a subset of $\{0,1/2,1\}$. This proves the lemma.
\end{proof}

\begin{proof}[Proof of Proposition \ref{dissoc}]
We only need to prove the inequality $\geq$. We first note that in the definition of
$\tilde e(A)$ we can restrict ourselves to finite sums. Indeed, for any $K\in\N$ put
$\tilde A_K= \sum_{k=K}^\infty A_k$, so that
$$
\sum_{k=1}^\infty \frac{A_k}{A}\  e^{(ball)}(A_k)
\geq \left( \sum_{k=1}^{K-1} \frac{A_k}{A}\  e^{(ball)}(A_k) + \frac{\tilde A_K}{A}
e^{(ball)}(\tilde A_K) \right) - \frac{\tilde A_K}{A} e^{(ball)}(\tilde A_K) \,.
$$
Since $a\, e^{(ball)}(a)\to 0$ as $a\to 0$, the last term goes to zero as
$K\to\infty$. The first term can be bounded from below by the claimed expression if we prove the result for finite sums, as we will do now.

In fact, we will prove that
\begin{equation}
\label{eq:dissocproof}
\inf \left\{ \sum_{k=1}^K \frac{A_k}{A}\  e^{(ball)}(A_k) :\ \sum_{k=1}^K A_k = A
\right\} = \min_{K\geq k\in\N} e^{(ball)}(A/k)
\end{equation}
by induction over $K\in\N$. For $K=1$ \eqref{eq:dissocproof} is clearly true. Now assume that
$K\geq 2$ and that the assertion is proved for $1,\ldots,K-1$. The set
$\{(A_1,\ldots,A_K)\in [0,A]^K:\ \sum A_k = A\}$ is compact and the function
to be minimized on the left side of \eqref{eq:dissocproof} is continuous
(note again that $a\, e^{(ball)}(a)$ is continuous at $a=0$), so there is a
minimizer $A^{(0)}$. If $A^{(0)}_k=0$ for some $k$, then
\eqref{eq:dissocproof} follows by induction assumption. Therefore we may
assume that $A^{(0)}_k\neq 0$ for all $k$. Let $1\leq i<j\leq K$ be
arbitrary. We shall show that $A^{(0)}_i = A^{(0)}_j$, which will prove
\eqref{eq:dissocproof}. The contribution from these two numbers to the left
side in \eqref{eq:dissocproof} is
\begin{align*}
\frac{A_i^{(0)}}{A} e^{(ball)}(A^{(0)}_i) + \frac{A_j^{(0)}}{A}
e^{(ball)}(A^{(0)}_j) = & \left( \theta^{2/3} + (1-\theta)^{2/3} \right) \left(
A^{(0)}_i + A^{(0)}_j \right)^{2/3} \frac{P(B)}{A\, |B|^{2/3}} \\
& + \left( \theta^{5/3} + (1-\theta)^{5/3} \right) \left( A^{(0)}_i + A^{(0)}_j
\right)^{5/3} \frac{D(B)}{A\, |B|^{5/3}}
\end{align*}
with $\theta = A^{(0)}_i/(A^{(0)}_i + A^{(0)}_j)$. By construction we have
$\theta\neq 0,1$. If we had $\theta\neq 1/2$, then by Lemma \ref{dissoclemma} this
contribution would be strictly less than for $\theta=1/2$. Thus, replacing
$A^{(0)}_i$ and $A^{(0)}_j$ both by $(A^{(0)}_i+A^{(0)}_j)/2$ we would get a
strictly smaller energy while preserving the constraint, but this contradicts the
minimality of $A^{(0)}$. Thus, we have $\theta=1/2$, that is, $A^{(0)}_i=A^{(0)}_j$,
which is what we wanted to prove.
\end{proof}


\bibliographystyle{amsalpha}

\end{document}